\documentclass[conference,a4paper]{IEEEtran}

\usepackage[utf8]{inputenc} 
\usepackage[T1]{fontenc}
\usepackage{url}              % provides \url{...}
\usepackage{cite}             % improves presentation of citations
\usepackage{hyperref}
\usepackage[cmex10]{amsmath}  % Use the [cmex10] option to ensure complicance
\usepackage{cleveref}
\crefname{algocf}{alg.}{algs.}
\Crefname{algocf}{Algorithm}{Algorithms}
\usepackage{mleftright}       % fix to wrong spacing of \left-,
\mleftright                   % \middle- \right-commands 

\usepackage{graphicx}         % provides \includegraphics{...} to
\usepackage{booktabs}         % fixes poor spacing in tables and
\usepackage[mode=buildnew]{standalone}

\usepackage[eulergreek]{sansmath}
\usepackage{pgfplots}
\usepackage{standalone}

\usepackage{verbatim}

\pgfplotsset{compat=1.16}
\pgfplotsset{
  tick label style={font=\footnotesize\sansmath\sffamily},
  every axis plot/.style={mark=none,line width=1pt},
}

\pgfplotsset{select coords between index/.style 2 args={
    x filter/.code={
        \ifnum\coordindex<#1\def\pgfmathresult{}\fi
        \ifnum\coordindex>#2\def\pgfmathresult{}\fi
    }
}}

\usepackage{amssymb}
\usepackage{amsthm}
\usepackage{stmaryrd}
\usepackage{xcolor}
\usepackage[ruled, vlined]{algorithm2e}
\usepackage[makeroom]{cancel}

\newtheorem{theorem}{Theorem}[section]

\newtheorem{lemma}{Lemma}[section]

\usepackage{amsmath,amsfonts,bm}

\def\eqref#1{equation~\ref{#1}}
\def\1{\bm{1}}

\DeclareMathAlphabet{\mathsfit}{\encodingdefault}{\sfdefault}{m}{sl}
\SetMathAlphabet{\mathsfit}{bold}{\encodingdefault}{\sfdefault}{bx}{n}

\newcommand{\KL}{D_{\mathrm{KL}}}

\DeclareMathOperator*{\argmin}{arg\,min}

\usepackage{mathtools}

\newcommand{\Nats}{\mathbb{N}}
\newcommand{\Ints}{\mathbb{Z}}
\newcommand{\Reals}{\mathbb{R}}
\newcommand{\Oh}{\mathcal{O}}

\newcommand{\Exp}{\mathbb{E}}
\newcommand{\Prob}{\mathbb{P}}

\newcommand{\Normal}{\mathcal{N}}

\newcommand{\Unif}{\mathcal{U}}

\newcommand{\Ind}{\mathbf{1}}
\newcommand{\defeq}{\stackrel{\mathit{def}}{=}}
\newcommand{\logtwo}{\log_2}
\newcommand{\loge}{\ln}
\newcommand{\BP}{\mathcal{BP}}
\newcommand{\lvlset}{\mathcal{H}}

\DeclarePairedDelimiterX{\infdivx}[2]{[}{]}{%
  #1\delimsize\| #2%
}
\DeclarePairedDelimiterX{\infdivxcolon}[2]{[}{]}{%
  #1\delimsize: #2%
}
\newcommand{\KLD}{\KL\infdivx}
\newcommand{\infD}{D_{\infty}\infdivx}

\DeclarePairedDelimiter{\norm}{\lVert}{\rVert}
\DeclarePairedDelimiter{\abs}{\lvert}{\rvert}
\DeclarePairedDelimiter{\clipped}{\llbracket}{\rrbracket}
\DeclarePairedDelimiterX{\innerProd}[2]{\langle}{\rangle}{%
    #1,#2%
}
\DeclareMathOperator*{\esssup}{ess\,sup}

\newcommand{\TargetFamily}{\mathcal{Q}}
\newcommand{\Entropy}{\mathbb{H}}

\newcommand{\disteq}{\stackrel{d}{=}}
\newcommand{\B}{\mathcal{B}}

\DeclareMathOperator\supp{supp}

\usetikzlibrary{arrows.meta, calc, positioning}

\newcommand{\offset}{1/16}
\newcommand{\vsamp}{0.1}

\tikzset{every picture/.style={/utils/exec={\sffamily}}}

\tikzstyle{arrow} = [thick,-latex]
\tikzstyle{message} = [rectangle, anchor=south west, draw=black!100, fill=black!10, align=left]

\newcommand{\appendixref}[1]{#1}

\begin{document}

\title{Adaptive Greedy Rejection Sampling} 

%%%%%%
\author{%
  \IEEEauthorblockN{Gergely Flamich}
  \IEEEauthorblockA{University of Cambridge \\
  Cambridge, UK \\
  gf332@cam.ac.uk}
  \and
  \IEEEauthorblockN{Lucas Theis}
  \IEEEauthorblockA{Google Research \\
  London, UK \\
  theis@google.com}
}

\maketitle

\begin{abstract}
We consider channel simulation protocols between two communicating parties, Alice and Bob. 
First, Alice receives a target distribution $Q$, unknown to Bob. 
Then, she employs a shared coding distribution $P$ to send the minimum amount of information to Bob so that he can simulate a single sample $X \sim Q$. 
For discrete distributions, Harsha et al. \cite{harsha2009communication} developed a well-known channel simulation protocol -- greedy rejection sampling (GRS) -- with a bound of ${\KLD{Q}{P} + 2\loge(\KLD{Q}{P} + 1) + \Oh(1)}$ on the expected codelength of the protocol. %--- break here, see originall version below
In this paper, we extend the definition of GRS to general probability spaces and allow it to adapt its proposal distribution after each step.
We call this new procedure Adaptive GRS (AGRS) and prove its correctness.
Furthermore, we prove the surprising result that the expected runtime of GRS is exactly $\exp(\infD{Q}{P})$, where $\infD{Q}{P}$ denotes the R\'enyi $\infty$-divergence. 
We then apply AGRS to Gaussian channel simulation problems. 
We show that the expected runtime of GRS is infinite when averaged over target distributions and propose a solution that trades off a slight increase in the coding cost for a finite runtime. 
Finally, we describe a specific instance of AGRS for 1D Gaussian channels inspired by hybrid coding \cite{theis2022algorithms}.
We conjecture and demonstrate empirically that the runtime of AGRS is $\Oh(\KLD{Q}{P})$ in this case.
\end{abstract}

\section{Introduction}
\noindent
Channel simulation is a two-party communication problem between Alice and Bob \cite{harsha2009communication}\cite{bennett2002}\cite{cuff2008}.
Before communication, they share a coding distribution $P$ and are assumed to have access to an infinite sequence of publicly available fair coin flips. 
During one round of communication, Alice receives a target distribution $Q$ and sends a number of bits to Bob so that he can simulate a single sample $X \sim Q$. We want the number of bits to be as small as possible. In the following, we will also be interested in the computational cost of generating the bits.
\par
\textit{One-shot channel simulation} (OSCS) recently began garnering interest outside of information theory, especially in the learned compression and differential privacy communities \cite{havasi2019minimal}\cite{flamich2020compressing}\cite{agustsson2020universally}\cite{shah2021optimal}\cite{flamich2022fast}\cite{theis2022lossy}. 
Many of the central problems in these fields, such as compressing data using probabilistic models \cite{flamich2020compressing}\cite{agustsson2020universally}\cite{theis2022lossy}, or compressing the models themselves \cite{havasi2019minimal}, can be cast as instances of OSCS.
However, the instances of OSCS that arise in these situations involve very high-dimensional distributions and can have high information content as measured by $\KLD{Q}{P}$.
Hence, it is crucial to understand the computational complexity of channel simulation protocols' encoding and decoding processes, not just their communication efficiency. 
Assuming RP $\neq$ NP, Agustsson \& Theis \cite{agustsson2020universally} showed that distributions exist for which the computational cost of encoding a sample grows at least as $\exp(\KLD{Q}{P})$. 
Yet, uniform distributions and one-dimensional Gaussian distributions can be simulated in $\Oh(\infD{Q}{P})$ or less  \cite{agustsson2020universally}\cite{flamich2022fast}.
\begin{algorithm}[t]
\SetAlgoLined
\DontPrintSemicolon
\SetKwInOut{Input}{Input}
\SetKwInOut{Output}{Output}
\SetKwFunction{clip}{clip}
\Input{Target $Q$, proposal $P$, bounds $(B_1, B_2, \dots)$}
$L_0, S_1 \gets (0, 1)$ \;
\For{$k = 1$ \KwTo $\infty$}{
$X_k \sim P\vert_{B_k}$ \;
$U_k \sim \Unif[0, 1]$ \;
$\beta_k \gets \clipped*{P(B_k) \cdot \left(\frac{dQ}{dP}(X_k) - L_{k - 1} \right) \middle / S_k}$ \;
 \If{$U_k \leq \beta_k$}{
    \KwRet{$X_k, k$}
 }
 $L_k \gets L_{k - 1} + S_k / P(B_k)$ \;
 $\lvlset_k \gets \left\{ y \in \Omega \mid L_k \leq \frac{dQ}{dP}(y) \right\}$ \;
 $S_{k + 1} \gets Q(\lvlset_k) - L_k \cdot P(\lvlset_k)$
}
\caption{Adaptive greedy rejection sampling. When $B_k = \Omega$ for every $k$, the algorithm is just greedy rejection sampling \cite{harsha2009communication}. See \Cref{sec:background} for notation.}
\label{alg:adaptive_grs}
\end{algorithm}
\par
\textit{Contributions.}
In this paper, we analyze greedy rejection sampling (GRS) \cite{harsha2009communication} and its extensions, with a particular focus on computational complexity. 
Our contributions are as follows:
\begin{itemize}
    \item We extend the formulation of GRS from discrete to arbitrary probability spaces and allow it to adapt to its target distribution.
    We call the generalized algorithm \textit{adaptive greedy rejection sampling} (AGRS) and use GRS to refer specifically to its non-adaptive version.
    We prove that AGRS retains the same correctness and codelength guarantees as Harsha et al.'s original algorithm \cite{harsha2009communication}.
    \item We prove that the expected runtime of GRS for a target $Q$ and proposal $P$ is exactly $\exp(\infD{Q}{P})$. Surprisingly, this matches the runtime of regular rejection sampling.
    \item We then consider a class of Gaussian channel simulation problems and show that the expected runtime of GRS for this class is infinite. To remedy this, we propose using an overdispersed proposal distribution, which trades off a slight increase in the coding cost for a finite, though still exponentially scaling, runtime. 
    \item Finally, we consider a particular instance of AGRS inspired by hybrid coding to simulate one-dimensional Gaussian channels. 
    We conjecture and demonstrate numerically that our proposed algorithm is correct and that it scales linearly in $\KLD{Q}{P}$ in the Gaussian case.
\end{itemize}

\section{Background}
\label{sec:background}
\noindent
\textit{Notation:} We denote the base 2 logarithm by $\logtwo$ and the natural logarithm by $\loge$.
We define $[a:b] \defeq [a, b] \cap \Ints$ for $a, b \in \Ints$, and $[n] \defeq [1:n]$ for $n \in \Nats$.
For $H \subseteq \Reals$ and $c \in \Reals$, define $c + H \defeq \{x + c \mid x \in H\}$.
Let $\clipped{\cdot}$ be the $[0, 1]$-clipping function: $\clipped{x} \defeq \min\{\max\{x, 0\}, 1\}$. 
Let $Q$ and $P$ be Borel probability measures over some Polish space $\Omega$ such that ${Q \ll P}$ with Radon-Nikodym derivative $dQ/dP$.
For a discrete random variable $K$, let $\Entropy[K]$ denote its Shannon entropy measured in nats.
Define ${\KLD{Q}{P} \defeq \int \loge \frac{dQ}{dP}(x) \, dQ(x)}$ and $\infD{Q}{P} \defeq \esssup_{x \in \Omega}\left\{\loge\frac{dQ}{dP}(x)\right\}$, where the essential supremum is taken with respect to $P$.
Throughout this paper we will assume that $\KLD{Q}{P} < \infty$.
For a positive measure $\mu$ over a measure space $\Omega$, define its total variation norm as $\norm{\mu}_{TV} \defeq \mu(\Omega)$.
For Borel sets $A$ and $B$, denote the measure of $A$ under $P$ restricted to $B$ as $P\vert_B(A) \defeq P(B)^{-1}\cdot\int_A \Ind[x \in B] \, dP(x)$, where $\Ind[\cdot]$ is the indicator function.
Let $\Unif(A)$ denote the uniform distribution over a set $A$. For example, \ $\Unif([0, 1))$ is uniform on the unit interval and $\Unif([n])$ is uniform over the first $n$ natural numbers.
Let $\Normal(\mu, \Sigma)$ be the $d$-dimensional Gaussian measure with mean $\mu$ and covariance $\Sigma$ and, as a slight abuse of notation, let $\Normal(x \mid \mu, \Sigma)$ denote its density with respect to the standard Lebesgue measure evaluated at $x$.
\par
\textit{Channel Simulation:}
Let Alice and Bob be communicating parties who share some coding distribution $P$ and have access to an infinite sequence of publicly available fair coin flips. 
Formally, the coin flips are a Bernoulli process $\BP = (b_1, b_2, \hdots)$ with mean $1/2$. 
In \textit{one-shot channel simulation} (OSCS), Alice receives a target distribution Q and needs to send the minimum information to Bob so that he can simulate a \textbf{single sample} $X \sim Q$.
In particular, Bob is not interested in learning $Q$.
If we make no further assumptions about $Q$ and $P$, this problem is also known as \textit{relative entropy coding} (REC) \cite{flamich2020compressing}\cite{flamich2022fast}.
On the other hand, suppose we can identify the target family with a class of conditional distributions $X \mid Z \sim Q_Z$ and have $Z \sim \pi$. In that case, this problem is also referred to as \textit{reverse channel coding} (RCC) \cite{agustsson2020universally}\cite{theis2022algorithms}.
Since $\pi$ defines a distribution over problem instances, we can perform an average-case analysis for RCC.
We use the term OSCS to refer to both REC and RCC.
We first obtain results for the REC setting, after which we analyze Gaussian RCC.
\par
\textit{Greedy Rejection Sampling (GRS):} 
GRS for discrete spaces was proposed by Harsha et al. \cite{harsha2009communication} and
is so-called because it accepts proposed samples as early as possible. 
That is, in each step $k$ it accepts candidate $X_k \sim P$ with the largest probability possible while still maintaining correctness. 
In other words, it greedily minimizes the survival probability ${S_k = \Prob[K \geq k]}$, where $K$ is the random variable corresponding to the step in which the sampler terminates.
GRS maintains an increasing sequence of levels $L_k$, which in turn define superlevel sets of the density ratio $\lvlset_k \defeq \left\{ y \in \Omega \mid L_k \leq \frac{dQ}{dP}(y) \right\}$.
Then, GRS accepts proposed samples $X_k$ that fall into $\lvlset_{k - 1}$ with high probability, and rejects all samples that fall outside.
\par
\textit{Channel Simulation with GRS:}
Alice can use GRS to solve the OSCS problem by using $\BP$ to simulate her proposed samples. 
She then sends the index $K$ of her accepted sample to Bob. 
Given $K$, Bob recovers $X_K$ by simulating the same samples Alice did using $\BP$, and accepting the $K$th one.
%Unfortunately, the distribution of $K$ is hard to work with generally. Hence, Alice usually uses Elias delta coding \cite{harsha2009communication} or the Zeta distribution \cite{li2018strong} to encode $K$.
In this paper, we generalize GRS in two ways: we extend it to arbitrary probability spaces and allow the proposal distributions to adapt to the target. 
We develop all our results for the generalized algorithm, after which we specialize them to GRS.

\section{Adaptive Greedy Rejection Sampling}
\noindent
In this section, we propose \textit{adaptive greedy rejection sampling} (AGRS) which allows the proposal distribution to change between iterations. 
The key observation is that from step $k$ onwards, GRS will never accept a sample outside of the superlevel set $\lvlset_{k - 1}$; hence we can adapt the proposal $P$ by truncating it in such a way that it still contains $\lvlset_{k - 1}$.
A similar observation is exploited by the hybrid coding scheme proposed by Theis \& Yosri \cite{theis2022algorithms}. 
However, hybrid coding only adapts to the target $Q$ once and is only applicable if $Q$ has bounded support. 
On the other hand, our scheme can adapt in every step and does not require $Q$ to have bounded support.
Let $B_1 \supseteq B_2 \supseteq \hdots$ be such that for all $k \in \Nats$ we have $\lvlset_{k - 1} \subseteq B_k$.
We refer to such a chain as the \textit{bounds used by the sampler}.
\Cref{alg:adaptive_grs} has pseudocode for AGRS.
By setting $B_k = \Omega$ for all $k$, we recover regular GRS.
To realize the bounds, Bob requires additional information from Alice. For certain shapes this information can be efficiently encoded with dithered quantization.
In \Cref{sec:gaussian_channel_simulation}, we elaborate on this approach on the example of 1D Gaussians.
%Unfortunately, finding a better chain of bounds appears to be heavily problem-specific. 
%Not only do the bounds need to satisfy the correctness criteria, but Alice also needs to be able to cheaply communicate them to Bob.
\par
Let $K$ now denote the random variable corresponding to the termination step of AGRS, and $X_K$ the sample it outputs. 
Then, our first result is that AGRS is correct, as it terminates with probability $1$ and its output follows the target distribution.
\begin{theorem}
\label{thm:adaptive_grs_correctness}
$\Prob[K < \infty] = 1 \quad \text{and} \quad X_K \sim Q$.
\end{theorem}
\begin{proof}
We first prove the following lemmas. Let $r = \frac{dQ}{dP}$.
\begin{lemma}
\label{lemma:survival_prob}
Let $S_k$ be as defined by \Cref{alg:adaptive_grs}. Then $S_k = \Prob[K \geq k]$.
\end{lemma}
\begin{proof}
We prove the statement via induction.
The base case $k = 1$ holds trivially, since $\Prob[K \geq 1] = 1 = S_1$ by definition.
Now, assume the statement holds for $k \leq n$. We need to show the statement holds for $k = n + 1$.
By definition, $\Prob[K = n \mid K \geq n, X_n = x_n] = \beta_n(x_n)$.
Multiplying both sides by $S_n = \Prob[K \geq n]$, taking expectations with respect to $X_n$ and noting that $\Prob[K = n, K \geq n] = \Prob[K = n]$, we find that
\begin{align}
    \Prob[K = n] 
    &= \int_{B_n} \beta_n(x) \frac{S_n}{P(B_n)} \, dP(x) \nonumber \\
    &= \int_{B_n} \Ind[L_{n - 1} < r(x) < L_n](r(x) - L_{n - 1}) \, dP(x) \nonumber \\
    &+ \int_{B_n}\Ind[L_n < r(x)]\frac{S_n}{P(B_n)} \, dP(x) \nonumber
\end{align}
\vspace{-0.5cm}
\begin{align}
    &= Q(\lvlset_{n-1} \setminus \lvlset_n) - L_{n - 1}P(\lvlset_{n-1} \setminus \lvlset_n) + \frac{S_n P(\lvlset_n)}{P(B_n)} \nonumber \\
    &=Q(\lvlset_{n-1}) - Q(\lvlset_n) - L_{n - 1}P(\lvlset_{n - 1}) + L_n P(\lvlset_n) \label{eq:accept_prob_telescopic} \\
    &= S_{n} - S_{n + 1}
    = P(K \geq n + 1) + P(K = n) - S_{n + 1}, \nonumber
\end{align}
where the third equality holds because $\lvlset_{n - 1} \subseteq B_n$.
Hence, $S_{n + 1} = \Prob[K \geq n + 1]$, finishing the proof.
\end{proof}
\begin{lemma}
Let $\Delta_k(A) \defeq Q(A) - \Prob[K \leq k, X_K \in A]$.
Then, $\Delta_k$ is a positive measure and $\lim_{k \to \infty}\norm{\Delta_k}_{TV} = 0$ implies \Cref{thm:adaptive_grs_correctness}.
\end{lemma}
\begin{proof}
Let $A$ be a Borel set.
Then, 
\begin{align}
    \Prob[X_K \in A] &= \sum_{k = 1}^\infty \Prob[K = k, X_k \in A] \nonumber \\
    &= \lim_{k \to \infty} \Prob[K \leq k, X_K \in A]. \nonumber
\end{align}
A computation analogous to \Cref{eq:accept_prob_telescopic} shows that 
${\Prob[K \leq k, X_K \in A] = Q(A) - Q(\lvlset_k \cap A) + L_k P(\lvlset_k \cap A)}$.
Thus, we find that $\Delta_k(A) = Q(\lvlset_k \cap A) - L_k P(\lvlset_k \cap A)$, which is positive by the definition of $\lvlset_k$.
Since $\Delta_k$ is positive, we have $\norm{\Delta_k}_{TV} = \Delta_k(\Omega) = 1 - \Prob[K \leq k] = S_{k + 1}$. Thus, $\lim_{k \to \infty}\norm{\Delta_k}_{TV} = 0$ implies $\Prob[K < \infty] = 1$.
Furthermore, $\norm{Q(A) - \Prob[X_K \in A]}_{TV} = \lim_{k \to \infty} \norm{\Delta_k}_{TV} = 0$ implies that the output distribution is $Q$, as required. 
\end{proof}
\begin{lemma}
\label{lemma:survival_prob_bound}
$S_{k + 1} \leq \exp\left(-\sum_{n = 1}^k \frac{P(\lvlset_n)}{P(B_n)} \right)$.
\end{lemma}
\begin{proof}
By the definitions of $S_{k + 1}$ and $L_k$,
\begin{align}
    S_{k + 1} &= 
    Q(\lvlset_k) - L_k P(\lvlset_k) \nonumber \\
    &= 
    Q(\lvlset_k) - L_{k - 1} P(\lvlset_k) - \frac{S_k P(\lvlset_k)}{P(B_k)} \nonumber \\
    &\leq Q(\lvlset_{k - 1}) - L_{k - 1} P(\lvlset_{k - 1}) - \frac{S_k P(\lvlset_k)}{P(B_k)} \nonumber\\
    &\leq S_k \left(1 - \frac{P(\lvlset_k)}{P(B_k)}\right). \nonumber
\end{align}
Applying the inequality $k$ times and noting that $S_1 = 1$, we find $S_{k + 1} \leq \prod_{n = 1}^k\left( 1 - \frac{P(\lvlset_n)}{P(B_n)}\right)$.
For $0 \leq c \leq 1$ we have $1 - c \leq \exp(-c)$.
We can apply this inequality to each factor of the product as $0 \leq P(\lvlset_n)/P(B_n) \leq 1$, since $\lvlset_n \subseteq B_n$.
\end{proof}
To finish the proof, we examine the two possible behaviours of the sequence $P(\lvlset_k)/P(B_k)$.
\textbf{Case 1:} Assume $P(\lvlset_k)/P(B_k) \to 0$ as $k \to \infty$.
Since $\Delta_k \ll P\lvert_{B_k}$, this implies $\Delta_k(\lvlset_k) \to 0 $.
This further implies that $\norm{\Delta_k}_{TV} = \Delta_k(\Omega) \to 0$, since $\Delta_k$ is entirely supported on $\lvlset_k$.
\textbf{Case 2:}
Assume $P(\lvlset_k)/P(B_k) \not\to 0$ as $k \to \infty$.
Then, since it is a positive sequence, there exists $\epsilon > 0$ for all $k$, such that $P(\lvlset_k)/P(B_k) \geq \epsilon$.
But then by \Cref{lemma:survival_prob_bound} we have $\norm{\Delta_k}_{TV} = S_{k + 1} \leq \exp\left(-\sum_{n = 1}^k \frac{P(\lvlset_n)}{P(B_n)}\right) \leq \exp(-k\epsilon) \to 0$ as $k \to \infty$, which finishes the proof.
\end{proof}

\begin{theorem}
\label{thm:agrs_codelength}
Define 
$m_k \defeq \argmin_{n \in [k]} \left\{ S_n / P(B_n) \right\}$.
Then,
$\Exp[\loge K] \leq \KLD{Q}{P} + \Exp\left[ \loge P(B_{m_K})  \right] + 1 + \loge 2$.
In particular, $\Exp[\loge K] \leq \KLD{Q}{P} + 1 + \loge 2$ for GRS.
\end{theorem}
\begin{proof}
Define $\alpha_k(x) \defeq \frac{d\Prob[K = k, X_k = x]}{dP}$.
Then,
\begin{align}
    \Exp[\loge K] &= \Prob[K = 1] \loge 1 + \sum_{k = 2}^\infty \Prob[K = k] \loge k \nonumber \\
    &= \sum_{k = 2}^\infty \int_{B_k} \alpha_k(x) \loge k \, dP(x). \label{eq:log_K_exact_form}
\end{align}
Next, we develop an appropriate bound for the integrand.
Note, that $\alpha_k$ is supported on $\lvlset_{k - 1}$. 
Then, for $x \in \lvlset_{k - 1}$
\begin{align}
    r(x) &\geq (k - 1)\left(\frac{1}{k - 1}\sum_{n = 1}^{k - 1}\frac{S_n}{P(B_n)}\right) \nonumber \\
    &\geq (k - 1)\min_{n \in [k]}\left\{\frac{S_k}{P(B_k)} \right\}. \nonumber
\end{align}
Second, for $k \geq 2$ we have $\loge k \leq \loge(k - 1) + \loge 2$.
Hence,
\begin{align}
   \loge k \leq \loge r(x) - \loge(1 - T_{m_k - 1}) + \loge P(B_{m_k}) + \loge 2. \nonumber
\end{align}
Substituting this back into \Cref{eq:log_K_exact_form} and simplifying:
\begin{align}
    \Exp[\loge K] &\leq \KLD{Q}{P} + \Exp[\loge P(B_{m_K})] - \Exp\left[\loge S_{m_K} \right] + \loge 2 \nonumber \\
    &\leq \KLD{Q}{P} + \Exp[\loge P(B_{m_K})] + 1 + \loge 2, \nonumber
\end{align}
where we used the fact that $ \Exp\left[\loge S_{m_K} \right] \geq \Exp\left[\loge S_K \right]$ and a result from \cite{harsha2009communication} that shows ${- \Exp\left[\loge S_K \right] \leq \int_0^1 \loge\left(\frac{1}{1 - p} \right) = 1}$.
This proves the first part of the theorem.
For the second part, we note that in GRS we set $B_k = \Omega$ for every $k$, from which $\Exp[\loge P(B_{m_K})] = 0$, and the result follows.
\end{proof}

Theorem~\ref{thm:agrs_codelength} can be used to show that in the RCC case, when $(X, Z) \sim P_{X, Z}$ and AGRS is applied to targets $P_{X \mid Z}$ and proposal $P_X$, we have
\begin{align}
    \Entropy[K] \leq C + \loge(C + 1) + 3.63
    \label{eq:agrs_coding_cost}
\end{align}
where $C = I[X; Z] + \mathbb{E}[\loge P(B_{m_K})]$, using a similar argument as Li \& El Gamal for their Theorem~1 in \cite{li2018strong}\appendixref{ (see Appendix~\ref{sec:entropy_bound})}.

% Here, the target distribution is assumed to depend on a random variable $\mu$ and we are considering the marginal entropy of $K$. The outer expectation is over $\mu$.

\begin{theorem}
\label{thm:adaptive_grs_runtime}
$ \Exp[K] \leq \exp \left ( \infD{Q}{P} \right)$, and we have equality in the case of GRS.
\end{theorem}
\begin{proof}
To begin, note that
\begin{equation}
\label{eq:runtime_upper_bound}
    \Exp[K] = \sum_{k = 1}^\infty \Prob[K \geq k] \leq \sum_{k = 1}^\infty S_k / P(B_k) = \lim_{k \to \infty} L_k.
\end{equation}
Next, we show the following lemma:
\begin{lemma}
\label{lemma:level_limit}
$\lim_{k \to \infty}L_k = \exp\left(\infD{Q}{P}\right)$.
\end{lemma}
\begin{proof}
For convenience, let $r^* \defeq \exp(\infD{Q}{P})$.
We show the statement by proving an upper and a lower bound on the limit.
For the upper bound we use an inductive argument.
For the base case, observe that $L_0 = 0 \leq r^*$.
Now assume $L_n \leq r^*$.
Then, for $L_{n + 1}$ we have
\begin{align}
    L_{n + 1} &= L_n + S_{n + 1} / P(B_{n + 1}) \nonumber \\
    &= L_n + P(B_{n + 1})^{-1} (Q(B_{n + 1}) - L_n P(B_{n + 1})) \nonumber \\
    &\leq r^*, \nonumber
\end{align}
where for the inequality we used the fact that 
\begin{align}
  Q(B_{n + 1}) = \int_{B_{n + 1}}r\, dP \leq \int_{B_{n + 1}}r^*\,dP = r^*P(B_{n + 1}). \nonumber
\end{align}
Since by the above argument $L_k \leq r^*$ for every $k$, the bound also must hold for the limit.
\par
We prove the lower bound via contradiction.
The sequence $L_k$ is monotone increasing and is bounded from above by $r^*$, hence it converges.
Assume now, that ${\lim_{k \to \infty} L_k = r^* - \epsilon}$ for some $\epsilon > 0$.
Define ${\mathcal{C} \defeq \left\{ x \in \Omega \mid r(x) \geq r^* - \frac{\epsilon}{2} \right\}}$.
Then,
\begin{align}
    S_{k + 1} &= Q(\lvlset_k) - L_k P(\lvlset_k) \nonumber\\
    &\geq Q(\mathcal{C}) - L_k P(\mathcal{C}) \nonumber\\
    &\geq \left(r^* - \frac{\epsilon}{2} - r^* + \epsilon \right) P(\mathcal{C}) \nonumber\\
    &= \frac{\epsilon}{2} P(\mathcal{C}), \nonumber
\end{align}
where for the second inequality we used the fact that ${Q(\mathcal{C}) \geq (r^* - \epsilon / 2)P(\mathcal{C})}$ by the definition of $\mathcal{C}$ and that ${L_k \leq r^* - \epsilon}$, since the $L_k$ sequence is monotonically increasing.
Note, that since ${r^* = \esssup r}$ with respect to $P$, we must have ${P(\mathcal{C}) > 0}$.
But now we find that for all $k$
\begin{align}
    L_{k + 1} - L_k = \frac{S_{k + 1}}{P(B_{k + 1})} \geq \frac{\epsilon}{2} \cdot \frac{P(\mathcal{C})}{P(B_{k + 1})} > 0, \nonumber
\end{align}
hence $L_k$ must diverge; a contradiction.
\end{proof}
Applying \Cref{lemma:level_limit} to \Cref{eq:runtime_upper_bound} yields the first part of the theorem.
For the second part, we note that for GRS \Cref{eq:runtime_upper_bound} holds with equality since $B_k = \Omega$ for all $k$.
\end{proof}

\section{Gaussian Channel Simulation}
\label{sec:gaussian_channel_simulation}
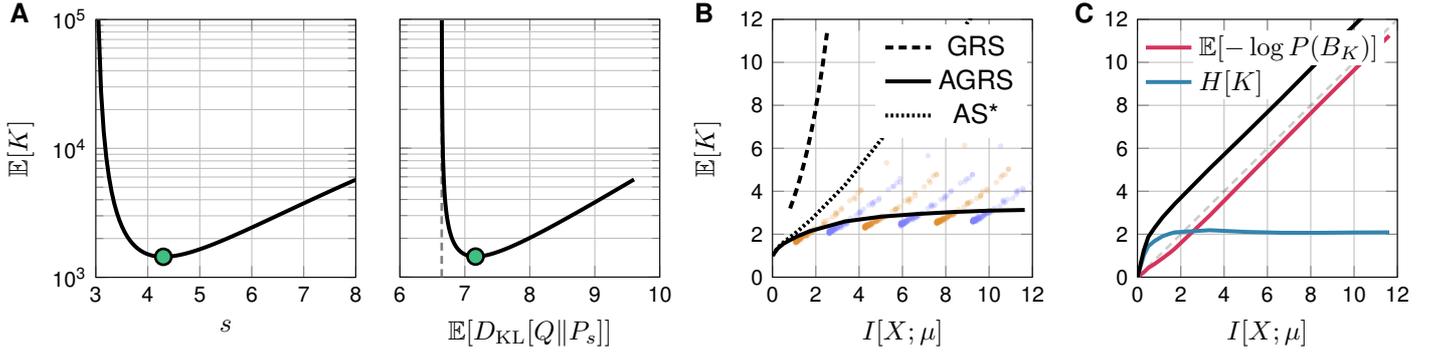
\begin{figure*}
    \hspace{-0.6cm}
\begin{tikzpicture}
    \node at (-1cm, 3.5cm) {\textbf{\textsf{A}}};
    \begin{axis}[
            xlabel={$s$},
            ylabel={$\mathbb{E}[K]$},
            width=5cm,
            height=5cm,
            xmin=3,
            xmax=8,
            xtick={3,4,...,8},
            ymin=1000,
            ymax=100000,
            grid=both,
            ymode=log,
        ]
        \addplot[black,line width=1.5pt] table[x=sigma,y=runtime,col sep=comma] {figures/data/runtime_4.csv};
        \addplot[white!50!blue!50!green,mark=*,mark size=2.3,only marks] coordinates {
            (4.299631726148781, 1441.99930652)
        };
        \addplot[black,mark=o,mark size=3.0,only marks] coordinates {
            (4.299631726148781, 1441.99930652)
        };
    \end{axis}
    \begin{axis}[
            xshift=4cm,
            xlabel={$\mathbb{E}[\KLD{Q}{P_s}]$},
            width=5cm,
            height=5cm,
            xmin=6,
            xmax=10,
            ymin=1000,
            ymax=100000,
            xtick={3,4,...,10},
            yticklabels={},
            grid=both,
            ymode=log,
        ]
        \addplot[black!50,densely dashed] coordinates {
            (6.643856189774725, 1000)
            (6.643856189774725, 100000)
        };
        \addplot[black,line width=1.5pt] table[x=bits,y=runtime,col sep=comma] {figures/data/runtime_4.csv};
        \addplot[white!50!blue!50!green,mark=*,mark size=2.3,only marks] coordinates {
            (7.164152419167847, 1441.99930652)
        };
        \addplot[black,mark=o,mark size=3.0,only marks] coordinates {
            (7.164152419167847, 1441.99930652)
        };
    \end{axis}
    \tikzstyle{dots}=[mark=star,mark size=1,opacity=.2,only marks]
    \tikzstyle{dots1}=[dots,yellow!50!red!90!black]
    \tikzstyle{dots2}=[dots,blue!50!white]
    \node at (8cm, 3.5cm) {\textbf{\textsf{B}}};
    \begin{axis}[
            xshift=8.9cm,
            xlabel={$I[X; \mu]$},
            ylabel={$\mathbb{E}[K]$},
            width=5cm,
            height=5cm,
            xmin=0,
            xmax=12,
            xtick={0,2,...,12},
            ytick={0,2,...,12},
            ymin=0,
            ymax=12,
            grid=both,
            %ymode=log,
            clip marker paths=true,
            legend style={draw=none},
        ]
        \addplot[dots1] table[x=kld,y=num_iter,col sep=comma,forget plot] {figures/data/kld_num_iter_6.csv};
        
        \addplot[dots2] table[x=kld,y=num_iter,col sep=comma,forget plot] {figures/data/kld_num_iter_7.csv};
        
        \addplot[dots1] table[x=kld,y=num_iter,col sep=comma,forget plot] {figures/data/kld_num_iter_8.csv};
        
        \addplot[dots2] table[x=kld,y=num_iter,col sep=comma,forget plot] {figures/data/kld_num_iter_9.csv};
        
        \addplot[dots1] table[x=kld,y=num_iter,col sep=comma,forget plot] {figures/data/kld_num_iter_10.csv};
        
        \addplot[dots2] table[x=kld,y=num_iter,col sep=comma,forget plot] {figures/data/kld_num_iter_11.csv};
        
        \addplot[densely dashed,black,line width=1.5pt] table[x=info,y=runtime,col sep=comma,select coords between index={1}{11}] {figures/data/runtime_grs.csv};
        \addlegendentry{\sffamily GRS};
        
        \addplot[black,line width=1.5pt] table[x=kld,y=num_iter,col sep=comma] {figures/data/mean_kld_num_iter.csv};
        \addlegendentry{\sffamily AGRS};
        
        \addplot[densely dotted,black,line width=1.5pt] table[x=mi,y=a_star_avg_runtime,col sep=comma] {figures/data/split_on_sample_a_star_avg_runtime_vs_mutual_info.csv};
        \addlegendentry{\sffamily AS*};
    \end{axis}
    \node at (13cm, 3.5cm) {\textbf{\textsf{C}}};
    \begin{axis}[
            xshift=13.7cm,
            xlabel={$I[X; \mu]$},
            width=5cm,
            height=5cm,
            xmin=0,
            xmax=12,
            ymin=0,
            ymax=12,
            xtick={0,2,...,12},
            ytick={0,2,...,12},
            grid=both,
            clip marker paths=true,
            legend style={draw=none,fill=none},
            legend cell align={left},
        ]
        \addplot[black,densely dashed,opacity=0.2,forget plot] coordinates {
            (0, 0)
            (20, 20)
        };
        
        \addplot[red!80!blue!80!white,line width=1.5pt] table[x=info,y=bound,col sep=comma] {figures/data/coding_cost.csv};
        \addlegendentry[fill=white]{$\mathbb{E}[-\log P(B_K)]$};
        
        \addplot[blue!60!green!80!white,line width=1.5pt] table[x=info,y=index,col sep=comma] {figures/data/coding_cost.csv};
        \addlegendentry[fill=white]{$H[K]$};
        
        \addplot[black,line width=1.5pt] table[x=info,y=sum,col sep=comma] {figures/data/coding_cost.csv};
    \end{axis}
\end{tikzpicture}
    \vspace{-0.2cm}
    \caption{\textbf{A:} The effect of using an overdispersed proposal distribution on the runtime of simulating 4-dimensional Gaussians with GRS. The target distributions all have variance 1 and the mean varies with standard deviation $\sigma = 3$. The proposal is Gaussian with variance $s^2 + 1$. The optimal coding cost is achieved when $s = \sigma$ but the expected runtime becomes infinite. The plot on the right shows that choosing the optimal runtime (green dot) increases the coding cost only slightly. \textbf{B:} Average runtime of GRS compared to Adaptive GRS applied to a 1-dimensional Gaussian target distribution whose variance is 1 and whose mean varies with standard deviation $\sigma$. Here, each point on the horizontal axis corresponds to a different value of $\sigma$. For GRS, we used the optimally overdispersed proposal distribution. For AGRS, no overdispersion was necessary. Additionally, each cluster (yellow or blue) corresponds to a fixed $\sigma$ and each point to a target distribution where we plotted the KL divergence against the expected runtime (estimated over 400 runs of AGRS with different random seeds). 
    AS* refers to the runtime of AS* coding from \cite{flamich2022fast} averaged over $10^5$ runs.
    \textbf{C:} Analysis of the coding cost. The amount of information contained in $K$ appears surprisingly constant and most of the information is contained in the bounds $B_K$. The black line indicates the sum of the two terms.}
    \label{fig:agrs_emp_results}
    \vspace{-0.5cm}
\end{figure*}
\noindent
In this section, we study the computational complexity of using GRS to simulate Gaussian channels where the input to the channel is also Gaussian distributed.
Formally, let $X$ be a $d$-dimensional Gaussian random variable, such that
\begin{align}
    X \mid \mu \sim \Normal(\mu, \rho^2 I), \quad \mu \sim \Normal(0, \sigma^2I), \nonumber
\end{align}
hence, marginally, $X \sim \Normal(0, (\sigma^2 + \rho^2)I)$.
In this setting, Alice would first draw a mean $\mu$, and Bob would like to simulate a sample $X \mid \mu$.
We assume that Alice and Bob share the marginal of $X$ before communication.
Then, by \Cref{thm:adaptive_grs_runtime},
\begin{align}
\label{eq:gaussian_rec_equation}
    \Exp[K \mid \mu] = \left(\frac{\rho^2 + \sigma^2}{\rho^2} \right)^{d/2} \exp\left(\frac{\norm{\mu}^2}{2\sigma^2}\right),
\end{align}
Now, taking the expectation over $\mu \sim \Normal(0, \sigma^2I)$, we find
\begin{equation}
    \Exp[K] = \infty. \nonumber
\end{equation}
In essence, this occurs because the right-hand side of \Cref{eq:gaussian_rec_equation} increases too rapidly in $\norm{\mu}$. 
To ameliorate this problem, we use an overdispersed proposal distribution ${\Normal(m, (\rho^2 + s^2)I)}$ instead, with some mean $m$ and $s^2 > \sigma^2$.
Computing $\Exp[K \mid \mu]$ and differentiating with respect to $m$, we find that $m = 0$ is optimal for any $s^2$.
Plugging this back into the expression for $\Exp[K \mid \mu]$ and taking expectation with respect to $\mu$, we find 
\begin{align}
    \Exp_{\mu \sim \Normal(0, \sigma^2I)}\left[\Exp[K \mid \mu]\right] = \left(\frac{s^2}{s^2 - \sigma^2} \cdot \frac{\rho^2 + s^2}{\rho^2}\right)^{d/2}. \nonumber
\end{align}
Differentiating with respect to $s^2$, we find that the optimal overdispersed variance is $s^2_{opt} = \sigma^2 + \sigma\sqrt{\rho^2 + \sigma^2}$.
Panel A in \Cref{fig:agrs_emp_results} illustrates how the runtime and coding cost scale as we increase the overdispersion.
\par
\textit{Implementing GRS for Gaussians:}
Here we describe how GRS can be implemented efficiently on a computer for a proposal $P_s = \Normal(0, (\rho^2 + s^2)I)$ and target $Q = \Normal(\mu, \rho^2I)$.
A direct calculation shows that
\begin{align}
    r(x) \defeq \frac{dQ}{dP}(x) &= \zeta \cdot \Normal(x \mid \nu, \kappa^2I), \nonumber
\end{align}
where we defined
\begin{align}
    \nu &\defeq \mu \cdot \frac{s^2 + \rho^2}{s^2} \nonumber\\
    \kappa &\defeq \frac{(s^2 + \rho^2) \cdot \rho^2}{s^2} \nonumber\\
    \zeta &\defeq \left(\frac{s^2 + \rho^2}{s^2}\right)^d \cdot \frac{1}{\Normal(\mu \mid 0, s^2I)}. \nonumber
\end{align}
\par
Now, we turn our attention to $Q(\lvlset_k)$ and $P_s(\lvlset_k)$.
First, note that for $x \in \lvlset_k$, we have
\begin{align}
-\loge L_k &\geq -\loge r(x) \nonumber \\
    &= -\loge \zeta + \frac{d}{2}\loge 2\pi + \frac{d}{2}\loge \kappa^2 + \frac{\norm{x - \nu}^2}{2\kappa^2} \nonumber
\end{align}
After some rearranging, we get
\begin{align}
    &\norm{x - \nu}^2 \leq R^2_k \nonumber \\
    &R^2_k \defeq \kappa^2\left(-2\loge L_k + 2 \loge \zeta - d\loge 2\pi - d \loge \kappa^2 \right), \nonumber
\end{align}
i.e., the $\lvlset_k$ sets will be spheres with center $\nu$ and radius $R_k$.
Assume now that $X_k \sim P_s$. 
Then
\begin{align}
    P_s(\lvlset_k) 
    &= \Prob \left[ \norm{X_k - \nu}^2 \leq R^2_k \right] \nonumber \\
    &= \Prob\left[\chi^2 \left(d, \frac{\norm{\nu}^2}{\rho^2 + s^2} \right) \leq \frac{R_k^2}{\rho^2 + s^2} \right], \nonumber
\end{align}
where $\chi^2(d, \lambda^2)$ is a noncentral chi-square random variable with $d$ degrees of freedom and noncentrality parameter $\lambda^2$.
The CDF of $\chi^2(d, \lambda^2)$ is readily available in modern numerical statistics packages.
An analogous argument shows that
\begin{align}
    Q(\lvlset_k) = \Prob\left[\chi^2\left(d, \frac{\norm{\nu - \mu}^2}{\rho^2}\right) \leq \frac{R_k^2}{\rho^2} \right]. \nonumber
\end{align}
We note that the above can be extended to non-isotropic $(Q, P)$ pairs. 
In that case, the $\lvlset_k$ sets are ellipsoidal and $P_s(\lvlset_k)$ and $Q(\lvlset_k)$ can be expressed using the CDF of a generalized chi-square random variable instead.
\par
\textit{AGRS for Gaussian channels:}
While \Cref{alg:adaptive_grs} can adapt to the target by using different bounds at each step, it is hard to leverage this additional flexibility without making further assumptions.
To solve the OSCS problem, Bob needs to know at least some information about the bounds Alice used. 
Otherwise Bob cannot decode Alice's message. 
Here, we propose a scheme for one-dimensional Gaussian channel simulation.
We conjecture and experimentally verify its correctness and efficiency.
Our method is inspired by hybrid coding \cite{theis2022algorithms} and dithered quantization (DQ) \cite{ziv1985universal} and is potentially applicable to other distributions and multivariate problems.
\par
First, let $\TargetFamily$ denote the family of possible targets under consideration. In our case, $\TargetFamily = \{\Normal(\mu, \rho^2) \mid \mu \in \Reals \}$ are one-dimensional Gaussian distributions.
We propose the following set of bounds.
Let $\lvlset_k^0$ be the $k$th level set computed for $\Normal(0, \rho^2)$ ($\lvlset_0^0 \defeq \Omega$) and set $B_k = \Phi^{-1}(c_k + \Phi(\lvlset^0_{k - 1}))$ for some $c_k$ which depends on $\mu$. Here, $\Phi$ is the CDF of $P$, that is, the marginal distribution of the Gaussian sample, $\mathcal{N}(0, \sigma^2 + \rho^2)$.
To ensure the correctness of \Cref{alg:adaptive_grs}, we require that for every $Q \in \TargetFamily$, there exists a $c_k$ so that $\lvlset_{k - 1} \subseteq B_k$.
Indeed, it is not hard to check that this holds for $k = 0, 1$.
We also verified numerically for several thousand choices of $Q$ that the relation holds for $k \geq 2$.

For AGRS to be decodable, Alice and Bob need to be able to simulate the same $X_k \sim P\vert_{B_k}$.
Fortunately, when $\Phi(B_k)$ are intervals, they can use DQ along inverse transform sampling to achieve precisely this at a coding cost of $-\logtwo P(B_k) + O(1)$ bits.
This idea is essentially identical to how DQ is used for hybrid coding \cite{theis2022algorithms} except that there the bounds are only adapted once before sampling begins, while in our case the bounds can change at each step of the sampling procedure.
\par
We found empirically (\Cref{fig:agrs_emp_results}B) that using this sequence of bounds lead to an exponential speed-up of \Cref{alg:adaptive_grs}.
Based on our results, we conjecture that if our proposed AGRS variant is correct, then its runtime is $\Oh(\KLD{Q}{P})$.
\par
The cost of encoding a sample consists of the cost of encoding $K$ plus the cost of communicating $X_k \sim P\vert_{B_k}$. The latter can be encoded with a variant of DQ based on bits-back coding \cite{bamler2022understanding} at a cost close to $\Exp[-\loge P(B_k)]$\appendixref{ (Appendix~\ref{sec:bbq})}. Let $N$ be the quantized representation produced by DQ. Then the coding cost of AGRS is
\begin{align}
\Entropy&[K, N] = \Entropy[K] + \Entropy[N \mid K] \nonumber \\
&\leq I[X; \mu] + \Exp[\loge P(B_{m_K})] - \Exp[\loge P(B_K)] \nonumber \\
&\quad + \loge(I[X; \mu] + \Exp[\loge P(B_{m_K})] + 1) + \Oh(1). \nonumber
\end{align}
Here we implicitly conditioned on the shared source of randomness to reduce clutter. In practice, we found the coding cost to be closer to $I[X; \mu] + 2$ (Figure~\ref{fig:agrs_emp_results}C). Interestingly, we also find that the majority of the information is contained in $N$ and that $H[K]$ stays around $2$ regardless of the choice of~$\sigma^2$.

\section{Discussion and Future Work}
\noindent
In this paper, we generalised the greedy rejection sampling algorithm first proposed by Harsha et al. \cite{harsha2009communication} by extending it to arbitrary probability spaces and allowing it to adapt to its target distribution at each step.
We showed that for a $(Q, P)$ target-proposal pair, the runtime of AGRS is upper bounded by $\exp(\infD{Q}{P})$ with equality in the case of GRS.
We studied Gaussian RCC using GRS and proposed an overdispersion scheme to ensure finite expected runtime of the algorithm at the cost of a small increase in the coding cost.
We proposed an AGRS scheme for the 1D Gaussian case and conjectured and empirically verified its correctness and runtime.
\par
There are several open problems and avenues for future work.
First, the inequality ${\KLD{Q}{P} \leq \infD{Q}{P}}$ can be arbitrarily loose in general. Hence, we ask if there exists a general channel simulation algorithm whose runtime scales as $\Oh(\exp(\KLD{Q}{P}))$, or if $\exp(\infD{Q}{P})$ is optimal.
Second, we ask if, and under what conditions \Cref{thm:agrs_codelength} and \Cref{thm:adaptive_grs_runtime} can be tightened.
Finally, we leave the correctness and runtime proofs of our proposed variant of AGRS for one-dimensional Gaussians for future work.

\section*{Acknowledgments}
This work was done while Gergely Flamich was an intern at Google. The authors would like to thank Phil Chou for helpful feedback on the manuscript.

\bibliographystyle{IEEEtran}
\bibliography{references}

\appendices

\section{Entropy Bound in \texorpdfstring{\Cref{eq:agrs_coding_cost}}{Equation (3)}}
\label{sec:entropy_bound}
\noindent
We derive the bound in \Cref{eq:agrs_coding_cost}.
In the RCC case, when $(X, Z) \sim P_{X, Z}$ and we apply AGRS to target $P_{X \mid Z}$ and proposal $P_X$,
\Cref{thm:agrs_codelength} yields
\begin{align}
    \Exp[\loge K] &\leq \Exp_{Z \sim P_Z}\left[\KLD{P_{X \mid Z}}{P_X} + \Exp[\loge P(B_{m_K}) \mid Z] \right] \nonumber\\
    &\quad + 1 + \loge 2 \nonumber\\
    & =\underbrace{I[X; Z] + \Exp[\loge P(B_{m_K})]}_{\defeq C} + 1 + \loge 2. \nonumber
\end{align}
From \cite{li2018strong}, we know that
\begin{align}
    \Entropy_2[K] &\leq \Exp[\logtwo K] + \logtwo(\Exp[\logtwo K] + 1) + 1. \nonumber
\end{align}
Switching to nats and applying the above equality for $\Exp[\loge K]$:
\begin{align}
     \Entropy[K] &\leq \Exp[\loge K] + \loge(\Exp[\logtwo K] + 1) + \loge 2 \nonumber\\
    &=\Exp[\loge K] + \loge\left(\frac{\Exp[\loge K]}{\loge 2} + 1\right) + \loge 2 \nonumber\\
    &=\Exp[\loge K] + \loge\left(\Exp[\loge K] + \loge 2 \right) - \loge \loge 2 + \loge 2 \nonumber\\
    &\leq C + \loge\left(C + 1 + 2\loge 2 \right) + 1 + 2\loge 2 - \loge \loge 2 \nonumber\\
    &\leq C + \loge\left(C + 1 \right) + 1 + 2\loge 2 - \loge \loge 2 + \loge(1 + 2 \loge 2) \nonumber\\
    &< C + \loge\left(C + 1 \right) + 3.63. \nonumber
\end{align}

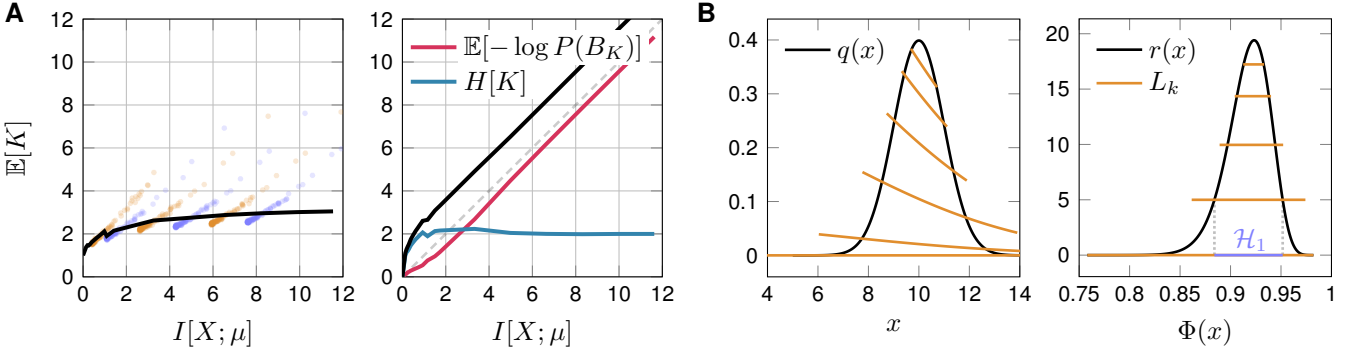
\begin{figure*}[t]
    \centering
    \hspace{-0.5cm}
\begin{tikzpicture}
    \node at (-0.9cm, 3.5cm) {\textbf{\textsf{A}}};
    \tikzstyle{dots}=[mark=star,mark size=1,opacity=.2,only marks]
    \tikzstyle{dots1}=[dots,yellow!50!red!90!black]
    \tikzstyle{dots2}=[dots,blue!50!white]
    \begin{axis}[
            xlabel={$I[X; \mu]$},
            ylabel={$\mathbb{E}[K]$},
            width=5cm,
            height=5cm,
            xmin=0,
            xmax=12,
            xtick={0,2,...,12},
            ytick={0,2,...,12},
            ymin=0,
            ymax=12,
            grid=both,
            %ymode=log,
            clip marker paths=true,
            legend style={draw=none},
        ]
        \addplot[dots1] table[x=kld,y=num_iter,col sep=comma,forget plot] {figures/data/non_fractional/kld_num_iter_12.csv};
        
        \addplot[dots2] table[x=kld,y=num_iter,col sep=comma,forget plot] {figures/data/non_fractional/kld_num_iter_15.csv};
        
        \addplot[dots1] table[x=kld,y=num_iter,col sep=comma,forget plot] {figures/data/non_fractional/kld_num_iter_16.csv};
        
        \addplot[dots2] table[x=kld,y=num_iter,col sep=comma,forget plot] {figures/data/non_fractional/kld_num_iter_17.csv};
        
        \addplot[dots1] table[x=kld,y=num_iter,col sep=comma,forget plot] {figures/data/non_fractional/kld_num_iter_18.csv};
        
        \addplot[dots2] table[x=kld,y=num_iter,col sep=comma,forget plot] {figures/data/non_fractional/kld_num_iter_19.csv};
        
       % \addplot[densely dashed,black,line width=1.5pt] table[x=info,y=runtime,col sep=comma,select coords between index={1}{11}] {figures/data/nonfractional/runtime_grs.csv};
       % \addlegendentry{GRS};
        
       \addplot[black,line width=1.5pt] table[x=kld,y=num_iter,col sep=comma] {figures/data/non_fractional/mean_kld_num_iter.csv};
       % \addlegendentry{AGRS};
    \end{axis}
    \begin{axis}[
            xshift=4.2cm,
            xlabel={$I[X; \mu]$},
            width=5cm,
            height=5cm,
            xmin=0,
            xmax=12,
            ymin=0,
            ymax=12,
            xtick={0,2,...,12},
            ytick={0,2,...,12},
            grid=both,
            clip marker paths=true,
            legend style={draw=none,fill=none},
            legend cell align={left},
        ]
        \addplot[black,densely dashed,opacity=0.2,forget plot] coordinates {
            (0, 0)
            (20, 20)
        };
        
        \addplot[red!80!blue!80!white,line width=1.5pt] table[x=info,y=bound,col sep=comma] {figures/data/non_fractional/coding_cost.csv};
        \addlegendentry[fill=white]{$\mathbb{E}[-\log P(B_K)]$};
        
        \addplot[blue!60!green!80!white,line width=1.5pt] table[x=info,y=index,col sep=comma] {figures/data/non_fractional/coding_cost.csv};
        \addlegendentry[fill=white]{$H[K]$};
        
        \addplot[black,line width=1.5pt] table[x=info,y=sum,col sep=comma] {figures/data/non_fractional/coding_cost.csv};
    \end{axis}
    \node at (8.2cm, 3.5cm) {\textbf{\textsf{B}}};
    \tikzstyle{ratio2}=[blue!50!white]
    \tikzstyle{proposal}=[yellow!50!red!90!black]
    \begin{axis}[
            xshift=9cm,
            legend style={draw=none,fill=none},
            legend cell align={left},
            legend pos={north west},
            xmin=4,
            xmax=14,
            xtick={4,6,...,14},
            ytick={0,0.1,...,0.5},
            width=4.9cm,
            height=5cm,
            xlabel={$x$},
        ]
        \addplot[black,line width=1pt] table[x=x,y=y,col sep=comma] {figures/data/agrs/q.csv};
        \addlegendentry{$q(x)$};
        
        \addplot[proposal,line width=1pt] table[x=x,y=y,col sep=comma] {figures/data/agrs/p_0.csv};
        
        \addplot[proposal,line width=1pt] table[x=x,y=y,col sep=comma] {figures/data/agrs/p_1.csv};
        \addplot[proposal,line width=1pt] table[x=x,y=y,col sep=comma] {figures/data/agrs/p_2.csv};
        \addplot[proposal,line width=1pt] table[x=x,y=y,col sep=comma] {figures/data/agrs/p_3.csv};
        \addplot[proposal,line width=1pt] table[x=x,y=y,col sep=comma] {figures/data/agrs/p_4.csv};
        \addplot[proposal,line width=1pt] table[x=x,y=y,col sep=comma] {figures/data/agrs/p_5.csv};
    \end{axis}
    \begin{axis}[
            xshift=13.1cm,
            legend style={draw=none,fill=none},
            legend cell align={left},
            legend pos={north west},
            xmin=0.75,
            xmax=1.0,
            xtick={0.75,0.8,...,1.05},
            ytick={0,5,...,25},
            width=4.9cm,
            height=5cm,
            xlabel={$\Phi(x)$},
        ]
        \addplot[gray!50!white, densely dotted, forget plot] coordinates {
            (0.88402, 0)
            (0.88402, 4.99868)
        };
        
        \addplot[gray!50!white, densely dotted, forget plot] coordinates {
            (0.95167, 0)
            (0.95167, 4.99868)
        };
        
        \addplot[black,line width=1pt] table[x=x,y=y,col sep=comma] {figures/data/agrs/r.csv};
        \addlegendentry{$r(x)$};
        
        \addplot[proposal,line width=1pt] table[x=x,y=y,col sep=comma] {figures/data/agrs/s_0.csv};
        \addlegendentry{$L_k$}
        
        \addplot[proposal,line width=1pt] table[x=x,y=y,col sep=comma] {figures/data/agrs/s_1.csv};
        \addplot[proposal,line width=1pt] table[x=x,y=y,col sep=comma] {figures/data/agrs/s_2.csv};
        \addplot[proposal,line width=1pt] table[x=x,y=y,col sep=comma] {figures/data/agrs/s_3.csv};
        \addplot[proposal,line width=1pt] table[x=x,y=y,col sep=comma] {figures/data/agrs/s_4.csv};
        \addplot[proposal,line width=1pt] table[x=x,y=y,col sep=comma] {figures/data/agrs/s_5.csv};
        
        \addplot[ratio2] coordinates {
            (0.88402, 0)
            (0.95167, 0)
        };
        
        \node[ratio2] at (axis cs:0.92,1.4) {$\lvlset_1$};
    \end{axis}
\end{tikzpicture}
    \caption{\textbf{A:} As in \Cref{fig:agrs_emp_results} but using wider bounds so that $1/P(B_k)$ is always an integer. This allows efficient communication of information about the bounds using dithered quantization without a need for bits-back coding. \textbf{B:} Illustration of AGRS on the example of a Gaussian distribution. On the right-hand side, $r = dQ/dP$ is plotted along with levels $L_k$ as a function of $\Phi(x)$, where $\Phi$ is the CDF of $P$. On the left-hand side, the same quantities are visualized as a function of $x$ and using the Lebesgue measure as a base measure. That is, $q = dQ/d\lambda$ and the yellow lines correspond to $L_k p(x)$ where $p = dP/d\lambda$. Additionally, the range of yellow lines is limited to $B_k$. Like GRS, AGRS slices the target distribution into a series of slices and targets one of the slices in each iteration.}
    \label{fig:agrs_illustration}
\end{figure*}

\section{Coding Cost of AGRS for Gaussian Channels}
Unlike in GRS, the index $K$ produced by AGRS is not sufficient to reconstruct the accepted candidate $X_K \sim P\vert_{B_K}$. For the case of the 1D Gaussian, we take
\begin{align}
    B_k = \Phi^{-1}(c_k + \Phi(\lvlset_{k - 1}^0)) \nonumber
\end{align}
and the receiver in general does not have access to $c_k$. Following Theis \& Yosri \cite{theis2022algorithms}, we assume that the missing information is communicated using dithered quantization (DQ). To this end, let $V_k \sim \Unif[-0.5, 0.5)$ be generated from a shared source of randomness $\BP$ and consider the following candidate generating process,
\begin{align}
    N_k &= \lfloor \textstyle\frac{c_k}{P(B_k)} - V_k \rceil, \nonumber\\
    Y_k &= (N_k + V_k) P(B_k), \nonumber\\
    X_k &= \Phi^{-1}(X_k), \nonumber
\end{align}
for all $k \in \mathbb{N}$. We assume that the sender chooses $c_k$ such that the support of $Y_k$ is contained in the interval $[0, 1)$.
While the receiver does not know $B_k$, we have $P(B_k) = |\Phi(\lvlset_k^0)|$ independent of the target distribution and so the receiver is able to compute $X_K$ after $K$ and $N_K$ have been received.
DQ has the property that $Y_k \sim c_k + P(B_k)V_k$ \cite{zamir2014lattice},
that is, $Y_k$ is uniform over the interval $c_k + \Phi(\lvlset_k^0)$ and so $X_k \sim P\vert_{B_k}$.
The sender transmits $K$ and $N_K$ which the receiver uses to reconstruct $X_K$. To reduce clutter, define
\begin{align}
    N &= N_K, & Y &= Y_K, & V &= V_K \nonumber
\end{align}
for the relevant quantities at acceptance.
Further, let ${Z = N + V}$ and let $p_{Z|K}$ be its density conditioned on $K$. Then
\begin{align}
    p_{Z | K}(z) 
    &= P_{N|K}(\lfloor z \rceil \mid K) p_{V|N,K}(z - \lfloor z \rceil \mid N = \lfloor z \rceil, K)
    \label{eq:z_to_nv}
\end{align}
and $p_{Z | K}(z) = p_{Y|K}(P(B_K)z) P(B_K)$.
We are interested in
\begin{align}
    H[(K, N) \mid \BP]
    &= H[K \mid \BP] + H[N \mid K, \BP] \nonumber \\
    &\leq H[K] + H[N \mid K, V]. \nonumber
\end{align}
For the second term, we have
\begin{align}
    & H[N \mid K, V] \nonumber\\
    &= h[(N, V) \mid K] - h[V \mid K] \nonumber\\
    &= \mathbb{E}[-\loge P(N \mid K) p(V \mid N, K) \mid K] - h[V \mid K] \nonumber\\
    &= \mathbb{E}[-\loge p_{Z|K}(Z)] - h[V \mid K] \label{eq:lotus} \\
    &= \mathbb{E}[-\loge p_{Y|K}(Y)] - \mathbb{E}[\loge P(B_K)] - h[V \mid K] \nonumber\\
    &= h[Y \mid K] - \mathbb{E}[\loge P(B_K)] - h[V \mid K] \nonumber\\
    &\leq \mathbb{E}[-\loge P(B_K)] - h[V \mid K] \nonumber
\end{align}
where we have applied the law of the unconscious statistician together with \Cref{eq:z_to_nv} in \Cref{eq:lotus} and the inequality follows because $Y$ is limited to the interval $[0, 1)$. Note that while the marginal distribution of $V_k$ is uniform over a unit interval for any fixed $k$, the marginal distribution of $V$ may be non-uniform if certain values of the dither are more likely to be accepted. Hence, in general we only have $h[V \mid K] \leq 0$. Theis \& Yosri \cite{theis2022algorithms} exploited that the accepted dither $V$ is marginally uniform when $1/P(B_k)$ are integers. In the following section, we propose a modification to DQ which produces uniform $V$ even when $1/P(B_k)$ are not integers.
For the 1D Gaussian case, we find that good performance can be achieved without BBQ by increasing intervals so that $1/P(B_k)$ is always an integer (\Cref{fig:agrs_emp_results}). However, we expect that this approach does not scale as well to multivariate problems as AGRS with BBQ.

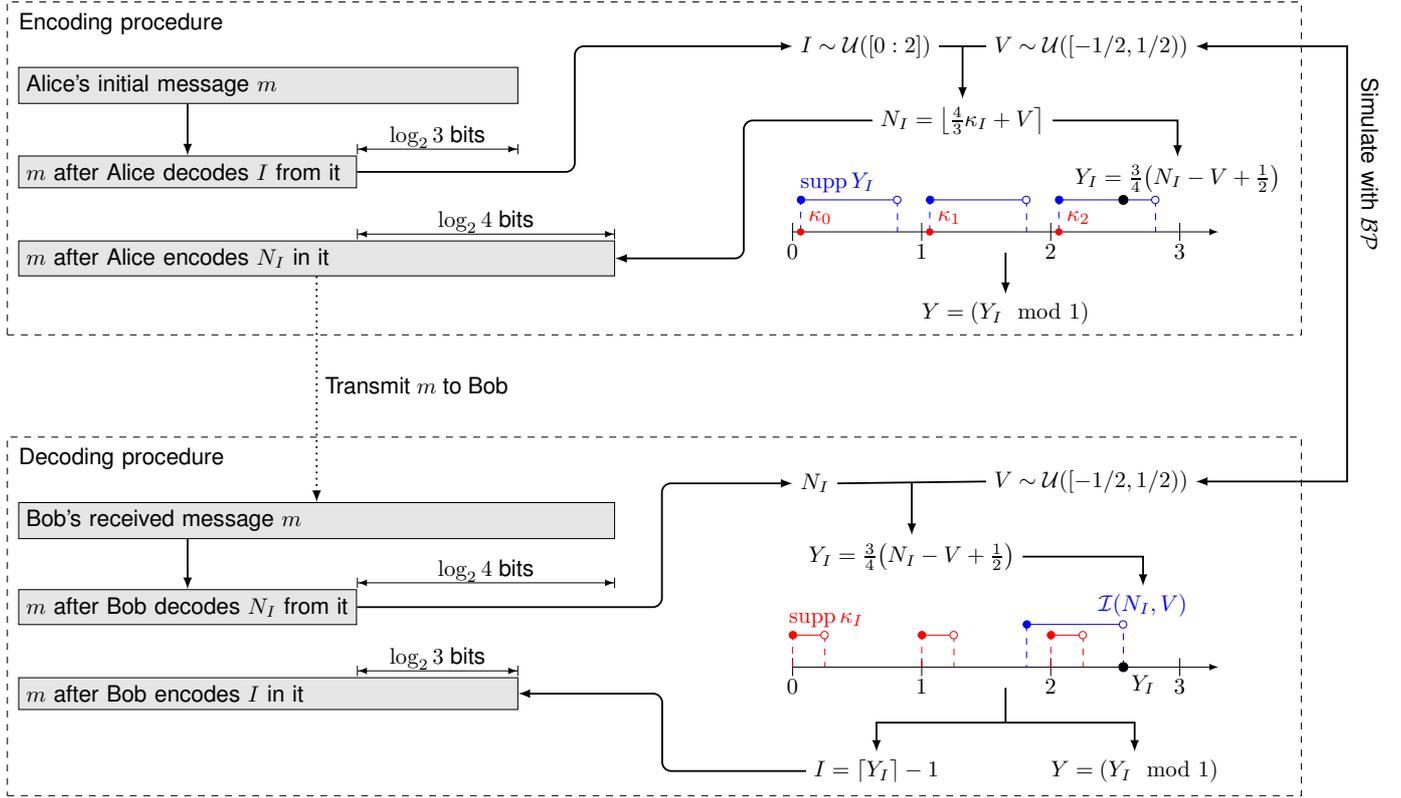
\begin{figure*}[t]
    \centering
    \begin{tikzpicture}[scale=2]

% ==========================
% Encoding
% ==========================

\begin{scope}
\node (encoding) [rectangle, anchor=north west, draw=black!100, dashed, text width=197mm, text depth=46mm, inner sep=5pt, xshift=-5pt] at (0, 1.8) {Encoding procedure\vfill};

\begin{scope}[shift={(6, 0)}]
% axis itself
\draw[-latex] (0, 0) -- (3.3, 0);

\foreach \x in {0, 1, 2, 3} {
% axis ticks and tick labels
\draw[] (\x, -2pt) -- (\x, 2pt) node[below=6pt] {$\x$};

% % quantization grid points
% \filldraw[green] (3/4 * \x, 0) circle (0.7pt);
}
% replicated target
\foreach \x in {0, 1, 2} {

% target support
\draw[blue, dashed] (\x + \offset, 0) -- (\x + \offset, 1/4);
\draw[blue, dashed] (\x + \offset + 3/4, 1/4) -- (\x + \offset + 3/4, 0);
\draw[blue, {Circle[]}-{Circle[fill=white]}] (\x + \offset - 0.03, 1/4) -- (\x + \offset + 3/4 + 0.03, 1/4);

% offsets
\filldraw[red] (\x + \offset, 0) circle (0.7pt) node[above right=-0pt]{$\kappa_{\x}$};
}

\node[blue, above right] at (0, 1/4) {$\supp Y_I$};

% sample Y_I
\filldraw (2.5 + \offset, 1/4) circle(1pt) node[above right, xshift=-25pt] (alice yi) {$Y_I = \frac{3}{4}\left(N_I - V + \frac{1}{2} \right)$};

\node (sample i) [anchor=south west] at (0, 1.3) {$I \sim \Unif([0:2])$};

\node (alice v) [anchor=south west] at (1.5, 1.3) {$V \sim \Unif([-1/2, 1/2))$};

\node (alice ni) [below = of $ (sample i.east)!0.5!(alice v.west) $, yshift=5pt]  {$N_I = \left\lfloor\frac{4}{3}\kappa_I + V \right\rceil$};

\draw[thick] (sample i.east) -- (alice v.west);
\draw[arrow] ($ (sample i.east)!0.5!(alice v.west) $) -- (alice ni);
\draw[arrow] (alice ni.east) -|  (alice yi.north);

\node (alice number line) at (3.3 / 2, -0.1) {};
\node[below = of alice number line, yshift=10pt] (alice y) {$Y = (Y_I \mod 1)$};

\draw[arrow] (alice number line) -- (alice y);

\end{scope}

% illustration of Alice's message bit string
\node[message, 
      text width=75mm] (alice initial) at (0, 1) {Alice's initial message $m$}; 
      
\node[message, text width=50mm, below= of alice initial.west, anchor=west, yshift=-10pt] (alice decoded) {$m$ after Alice decodes $I$ from it}; 
      
\draw[|{latex}-{latex}|] ([yshift=0.5mm]alice decoded.north east) -- ([yshift=0.5mm]alice initial.east|-alice decoded.north) node[midway, above=-1mm]{$\log_2 3$ bits};
      
\node[message, text width=90mm, below = of alice decoded.west, anchor=west, yshift=-10pt] (alice encoded) {$m$ after Alice encodes $N_I$ in it};
      
\draw[|{latex}-{latex}|] ([yshift=0.5mm]alice decoded.east|-alice encoded.north) -- ([yshift=0.5mm]alice encoded.north east) node[midway, above=-1mm]{$\log_2 4$ bits};

\draw[arrow] (alice decoded.north|-alice initial.south) -- (alice decoded);
\draw[arrow, rounded corners] (alice decoded.east) -| ($ (alice decoded.east)!0.5!(sample i.west) $) |- (sample i.west);
\draw[arrow, rounded corners] (alice ni) -| ($ (alice ni.west)!0.5!(alice encoded.east) $) |- (alice encoded.east);

\end{scope}

% ==========================
% Decoding
% ==========================

\begin{scope}[shift={(0, -34mm)}]

\node (decoding) [rectangle, anchor=north west, draw=black!100, dashed, text width=197mm, text depth=50mm, inner sep=5pt, xshift=-5pt] at (0, 1.8) {Decoding procedure\vfill};

\node[message, 
      text width=90mm] (bob initial) at (0, 1) {Bob's received message $m$}; 
      
\node[message, text width=50mm, below= of bob initial.west, anchor=west, yshift=-10pt] (bob decoded) {$m$ after Bob decodes $N_I$ from it}; 
      
\draw[|{latex}-{latex}|] ([yshift=0.5mm]bob decoded.north east) -- ([yshift=0.5mm]bob initial.east|-bob decoded.north) node[midway, above=-1mm]{$\log_2 4$ bits};
      
\node[message, text width=75mm, below = of bob decoded.west, anchor=west, yshift=-10pt] (bob encoded) {$m$ after Bob encodes $I$ in it};

\draw[|{latex}-{latex}|] ([yshift=0.5mm]bob decoded.east|-bob encoded.north) -- ([yshift=0.5mm]bob encoded.north east) node[midway, above=-1mm]{$\log_2 3$ bits};

\begin{scope}[shift={(6, 0)}]
% axis itself
\draw[-latex] (0, 0) -- (3.3, 0);

\foreach \x in {0, 1, 2, 3} {
% axis ticks and tick labels
\draw[] (\x, -2pt) -- (\x, 2pt) node[below=6pt] {$\x$};

}
% replicated target
\foreach \x in {0, 1, 2} {

% kappa support
\draw[red, dashed] (\x, 0) -- (\x, 1/4);
\draw[red, dashed] (\x + 1/4, 1/4) -- (\x + 1/4, 0);
\draw[red, {Circle[]}-{Circle[fill=white]}] (\x - 0.03, 1/4) -- (\x + 1/4 + 0.03, 1/4);
}
\node[red, above right, xshift=-5pt] at (0, 1/4) {$\supp \kappa_I$};

% kappa_I support 
\draw[blue, dashed] (2.5 + \offset, 0) -- (2.5 + \offset, 1/3);
\draw[blue, dashed] (2.5 + \offset - 3/4, 1/3) -- (2.5 + \offset - 3/4, 0);
\draw[blue, {Circle[fill=white]}-{Circle[fill=blue]}] (2.5 + \offset + 0.03, 1/3) -- (2.5 + \offset - 3/4 - 0.03, 1/3) node[above right, xshift=30pt] (bob inferred interval) {$\mathcal{I}(N_I, V)$};

\filldraw (2.5 + \offset, 0) circle(1pt) node[below right] (bob yi two) {$Y_I$};

\node (bob ni) [anchor=south west] at (0, 1.3) {$N_I$};

\node (bob v) [anchor=south west] at (1.5, 1.3) {$V \sim \Unif([-1/2, 1/2))$};

\node (bob yi) [below = of $ (bob ni.east)!0.5!(bob v.west) $, yshift=5pt]  {$Y_I = \frac{3}{4}\left(N_I - V + \frac{1}{2} \right)$};

\draw[thick] (bob ni.east) -- (bob v.west);
\draw[arrow] ($ (bob ni.east)!0.5!(bob v.west) $) -- (bob yi);
\draw[arrow] (bob yi.east) -|  (bob inferred interval.north);

\node (bob number line) at (3.3 / 2, -0.1) {};
\node[below = of bob number line, yshift=13pt] (bob arrow mid) {};
\node[below = of bob number line, yshift=0pt, xshift=-20mm] (bob i) {$I = \lceil Y_I \rceil - 1$};
\node[below = of bob number line, yshift=0pt, xshift=20mm] (bob y) {$Y = (Y_I \mod 1)$};

\draw[arrow] (bob arrow mid.north) -| (bob i);
\draw[arrow] (bob arrow mid.north) -| (bob y);
\draw[thick] (bob number line) -- (bob arrow mid.north);

\end{scope}

\draw[arrow] (bob decoded.north|-bob initial.south) -- (bob decoded);
\draw[arrow, rounded corners] (bob decoded.east) -| ($ (bob decoded.east)!0.7!(bob ni.west) $) |- (bob ni.west);
\draw[arrow, rounded corners] (bob i) -| ($ (bob i.west)!0.5!(bob encoded.east) $) |- (bob encoded.east);

\end{scope}

\draw[arrow, dotted] (alice encoded.south) -- (bob initial.north) node[midway, above right, yshift=-2mm]{Transmit $m$ to Bob};

\draw[latex-latex, thick] (alice v.east) -| ([xshift=10pt]encoding.east) node[midway, right, sloped, xshift=10pt, yshift=10pt] {Simulate with $\BP$} |- (bob v.east) ;

\end{tikzpicture}
    \caption{Bits-back Quantization encoding and decoding procedure for a target distribution of $Q = \Unif(\kappa + [0, 3/4))$, i.e., \ $a = 3$ and $b = 4$.
    In the figure, $\supp Y_I$ and $\supp \kappa_I$ denote the support of $Y_I$ and $\kappa_I$, respectively.
    The encoding procedure begins in the top left with ``Alice's initial message $m$'' and follows the solid arrows.
    For a detailed description of the procedure and the notation, see \Cref{sec:bbq}.
    We assume that Alice already has some message $m$ she wishes to communicate to Bob.
    We also assume that both Alice and Bob can simulate $V$ with the publicly available randomness $\BP$ and they use the same invertible sampler to encode and decode messages to/from $m$.
    Crucially, since Alice begins by decoding $\log_2 3$ bits from her message, at the end of the encoding process her message length only increases by $-\log_2 3/4$ bits in total.}
    \label{fig:bbq_illustration}
\end{figure*}

\section{Bits-back Quantization}
\label{sec:bbq}
\noindent
We describe a channel simulation algorithm for a proposal $P$ and target ${P\vert_{B}}$ supported on $\Reals$ with asymptotic coding cost $-\log P(B) + \Oh(1)$, that we call bits-back quantization (BBQ).
For a graphical illustration of the procedure, see \Cref{fig:bbq_illustration}.
We assume, that $P$ has invertible CDF $\Phi$ and $B = (u, v)$ is an interval, where $u, v \in \Reals$.
We assume that $P(B) = a / b$ with $a, b \in \Ints^+$ and $a \leq b$.

\par
The problem can be reduced to sampling a uniform target using a uniform proposal, as follows.
Let 
\begin{align}
    \B \defeq \Phi(B) = \{ \Phi(x) \mid x \in B \}. \nonumber
\end{align}
WLOG, we may assume that $\B$ can be decomposed as
\begin{align}
    \B = \kappa + [0, a / b), \quad \kappa \in [0, 1 - a/b]. \nonumber
\end{align}
Since $P(B) = \Phi(v) - \Phi(u)$, the above decomposition ensures that $\abs{\B} = P(B) = a/b$, while the restriction on the range of the offset $\kappa$ ensures that $\B \subseteq [0, 1)$.
Let $X \sim P\vert_{B}$ and $Y \sim \Unif(\kappa + [0, a/b))$.
It is a standard result from inverse transform sampling that
\begin{align}
    \Phi^{-1}(Y) \disteq X, \nonumber
\end{align}
where $\disteq$ denotes equality in distribution.
Therefore, we restrict our attention to encoding $Y$, which can be transformed via $\Phi^{-1}$ to follow the desired truncated distribution.
\par
We begin by describing the special case of the problem when $a = 1$.
In this case, we use DQ directly to encode $Y$.
Recall, that for $V, V' \sim \Unif([-1/2, 1/2))$ we have \cite{ziv1985universal}\cite{zamir2014lattice}:
\begin{align}
    \lfloor c + V \rceil - V \disteq c + V', \nonumber
\end{align}
where we always round towards $\infty$, i.e., ${\lfloor x \rceil \defeq \lfloor x + 1/2 \rfloor}$ for $x \in \Reals$.
Let us assume that Alice and Bob simulate $V$ using their shared source of randomness.
Then, let
\begin{align}
    N &\defeq \left\lfloor b \cdot \kappa + V \right\rceil \nonumber \\
    Y &\defeq \frac{1}{b} \cdot \left(N - V + \frac{1}{2}\right). \nonumber
\end{align}
Then, $Y \sim \Unif(\kappa + [0, 1/b))$ as desired.
Since by assumption $\kappa \in [0, 1 - 1/b)$, we have $N \in [0:b - 1]$, thus Alice can communicate $N$ in $\lfloor \logtwo b \rfloor \leq \logtwo b + 1 = -\logtwo P(B) + 1$ bits.
Further, whenever $Y$ is uniformly distributed over $[0, 1)$, then $V = bY - N$ must be uniformly distributed over $[0, 1)$.

\par
We now consider an extension to the above scheme for general $a \in \Ints^+, a < b$, in which case $Y \sim \Unif(\kappa + [0, a/b))$.
Or proposal is inspired by bits-back coding \cite{bamler2022understanding}.
First, Alice replicates the unit interval $a$ times to obtain $[0, a)$.
She also creates $a$ copies of the offset, $\kappa_i \defeq i + \kappa$ for $i \in [0:a - 1]$.
Consider encoding the sample redundantly $a$ times. If
\begin{align}
    N_i &\defeq \left\lfloor \frac{b}{a} \cdot \kappa_i + V \right\rceil \nonumber \\
    Y_i &\defeq \frac{a}{b}\cdot\left(N_i - V + \frac{1}{2}\right) \nonumber
\end{align}
then
\begin{align}
   \forall i \in [0:a - 1], \quad Y \disteq (Y_i \mod 1). \nonumber
\end{align}
The significance of the above is that Alice has freedom in choosing $\kappa_i$ and she has no a priori preference.
Hence, we suggest that before encoding $N$, she first simulates ${I \sim \Unif([0:a-1])}$ using an invertible sampler by decoding $\logtwo a$ bits from the message she plans to transmit to Bob.
We can use the algorithm in Bamler's Listing 1 \cite{bamler2022understanding} for this.
Then, Alice transmits $N_I$ to Bob.
Note, that
\begin{align}
    \supp\kappa_I = \bigcup_{j = 0}^{a - 1}\left[j, j + 1 - \frac{a}{b} \right], \label{eq:support_of_kappa_I}
\end{align}
where $\supp \kappa_I$ denotes the support of $\kappa_I$.
Hence, similarly to the $a = 1$ case, $N_I$ can take $b$ different values.
Thus, Alice requires $\logtwo b$ bits to encode $N_I$.
While Bob technically receives $\logtwo b$ bits from Alice, he can recover the $\logtwo a$ bits Alice used to simulate $I$, as we show next.
Once Bob decodes $N_I$, he computes $Y_I$, from which he can compute $Y$.
However, he can also use $N_I$ to compute a range in which $\kappa_I$ must fall:
\begin{align}
&\left\lfloor \frac{b}{a} \cdot \kappa_I + V \right\rceil = N_I \nonumber \\
\Leftrightarrow \quad &\frac{a}{b}\cdot\left(N_I - V - \frac{1}{2}\right) \leq \kappa_I < \frac{a}{b}\cdot\left(N_I - V + \frac{1}{2}\right) \nonumber \\
\Leftrightarrow \quad &Y_I - \frac{a}{b} \leq \kappa_I < Y_I \nonumber
\end{align}
Let $\mathcal{I}(N_I, V)$ denote the above interval, and note that its width is $a / b$.
However, Bob has a second piece of information at his disposal, namely that $\kappa_I$ is supported on the set defined in \Cref{eq:support_of_kappa_I}.
Note, that for all $0 \leq j < a - 1$, the distance between $[j, j + 1 - a/b]$ and $[j + 1, j + 2 - a/b]$ is precisely $a/b$, meaning that $\mathcal{I}(N_I, V)$ will intersect exactly one interval in $\supp \kappa_I$, namely $[I, I + 1 - a/b)$.
Two intervals intersect each other if and only if the supremum of both is greater than the infimum of the other, i.e.\
\begin{align}
    &\left[I, I + 1 - \frac{a}{b}\right] \cap \left[Y_I - \frac{a}{b}, Y_I \right) \neq \emptyset \nonumber\\
    \Leftrightarrow\quad & I < Y_I \quad \text{and} \quad Y_I - \frac{a}{b} \leq I + 1 - \frac{a}{b} \nonumber \\
    \Leftrightarrow\quad & I < Y_I \leq I + 1 \nonumber \\
    \Leftrightarrow\quad & I = \left\lceil Y_I\right\rceil - 1. \label{eq:index_computation_rule}
\end{align}
Thus, Bob can recover $I$ using \Cref{eq:index_computation_rule}.
He encodes $I$ back into his remaining message using the same invertible sampler Alice used to decode.
In total Bob receives $\logtwo b$ bits, but recovers $\logtwo a$ bits, thus the communication cost of this procedure is
\begin{align}
    \logtwo b - \logtwo a = -\logtwo\frac{a}{b} = -\log P(B), \nonumber
\end{align}
as required.
\par \textit{Note:} The boundary case $Y_I = I + 1$ occur almost never, so in practice it is also reasonable to use $I = \lfloor Y_I \rfloor$.
\par
Finally, if $Y_I$ is uniform over $[0, a)$, then $V = \frac{b}{a} Y_I - N_I$ must be uniform over $[0, 1)$.

\end{document}